\documentclass{article}

\usepackage{amsmath}
\usepackage{amsfonts}
\usepackage{amssymb}
\usepackage{amsthm}
\usepackage{enumerate}
\usepackage{bbm}
\usepackage{hyperref}

\newtheorem{theorem}{Theorem}[section]
\newtheorem{lemma}[theorem]{Lemma}
\newtheorem{claim}[theorem]{Claim}
\newtheorem{definition}[theorem]{Definition}
\newtheorem{corollary}[theorem]{Corollary}

\newtheorem{fact}[theorem]{Fact}
\newtheorem{proposition}[theorem]{Proposition}
\newtheorem{observation}[theorem]{Observation}

\newtheorem*{theorem*}{Theorem}
\newtheorem*{corollary*}{Corollary}

\newcommand{\suchthat}{\;\ifnum\currentgrouptype=16 \middle\fi|\;}

\newcommand{\F}{\mathbb{F}}
\newcommand{\C}{\mathbb{C}}
\newcommand{\R}{\mathbb{R}}
\newcommand{\N}{\mathbb{N}}
\newcommand{\Z}{\mathbb{Z}}
\newcommand{\Prob}{\mathbb{P}}

\DeclareMathOperator{\Ber}{Ber}

\DeclareMathOperator{\RS}{RS}
\newcommand{\E}{{\rm I\kern-.3em E}}

\begin{document}

\title{On the $\text{AC}^0[\oplus]$ complexity of Andreev's Problem}
\author{
Aditya Potukuchi  \thanks{Department of Computer Science,  Rutgers University. {\tt aditya.potukuchi@cs.rutgers.edu}. Research supported in part by NSF grant 
CCF-1514164} 
}

\maketitle

\begin{abstract}
Andreev's Problem states the following: Given an integer $d$ and a subset of $S \subseteq \F_q \times \F_q$, is there a polynomial $y = p(x)$ of degree at most $d$ such that  for every $a \in \F_q$, $(a,p(a)) \in S$? We show an $\text{AC}^0[\oplus]$ lower bound for this problem. 

This problem appears to be similar to the list recovery problem for degree $d$-Reed-Solomon codes over $\F_q$ which states the following: Given subsets $A_1,\ldots,A_q$ of $\F_q$, output all (if any) the Reed-Solomon codewords contained in $A_1\times \cdots \times A_q$. For our purpose, we study this problem when $A_1, \ldots, A_q$ are random subsets of a given size, which may be of independent interest.
\end{abstract}

\section{Introduction}

For a prime power $q$, let us denote by $\F_q$, the finite field of order $q$. Let us denote the elements of $\F_q = \{a_1,\ldots, a_q\}$. One can think of $a_1,\ldots,a_q$ as some ordering of the elements of $\F_q$. Let $\mathcal{P}_d = \mathcal{P}_d^{q}$ be the set of all univariate polynomials of degree at most $d$ over $\F_q$.  Let us define the problem which will be the main focus of this paper:

\begin{itemize}
\item[]\textbf{Input:} A subset $S \subset \F_q^2$, and integer $d$.
\item[]\textbf{Output:} Is there a $p \in \mathcal{P}_q^d$ such that $\{(a_i,p(a_i)) \suchthat i\in [q]\} \subseteq S$?
\end{itemize}

The problem of proving $\text{NP}$-hardness of the above function seems to have been first asked in~\cite{J86}. It was called `Andreev's Problem' and still remains open. One may observe that above problem is closely related to the \emph{List Recovery of Reed-Solomon codes}. In order to continue the discussion, we first define Reed-Solomon codes:

\begin{definition}[Reed-Solomon code]
\label{def:RS}
The degree $d$ Reed-Solomon over $\F_q$, abbreviated  as $\RS[q,d]$ is the following set:
\[
\RS[q,d] = \{(p(a_1),\ldots,p(a_q)) \suchthat p \in \mathcal{P}_{d}^q\}
\] 
\end{definition}

Reed-Solomon codes are one of the most widely (if not the most widely) studied families of error-correcting codes. It can be checked that $\RS[q,d]$ is a $d+1$-dimensional subspace of $\F_q^q$ such that every non-zero vector has at least $q - d$ non-zero coordinates. In coding theoretic language, we say that $\RS[q,d]$ is a \emph{linear code} of \emph{block length} $q$, \emph{dimension} $d+1$ and \emph{distance} $q - d$. The \emph{rate} of the code is given by $\frac{d}{q}$. The set of relevant facts about Reed-Solomon codes for this paper may be found in Appendix~\ref{sec:RS}

The \emph{List Recovery} problem for a code $\mathcal{C} \subset \F_q^n$ is defined as follows:

\begin{definition}[List Recovery problem for $\mathcal{C}$]
\begin{itemize}
\item[]
\item[]\textbf{Input:} Sets $A_1,\ldots, A_n \subseteq \F_q$.
\item[]\textbf{Output:} $\mathcal{C} \cap (A_1 \times\cdots \times A_n)$
\end{itemize}
\end{definition}

Given the way we have defined these problems, one can see that Andreev's Problem is essentially proving $\text{NP}$-hardness for the List Recovery of Reed-Solomon codes where one just has to output a Boolean answer to the question 
\[
\mathcal{C} \cap (A_1 \times\cdots \times A_n) \neq \emptyset?
\] 

Indeed, let us consider a List Recovery instance where the code $\mathcal{C}$ is $\RS[q,d]$, and the input sets are given by $A_1,\ldots,A_q$. Let us identify $(A_1,\ldots, A_q)$ with the set
\[
S = \bigcup _{i \in [q]}\{(a_i,z) \suchthat z \in A_i\} \subseteq \F_q^2
\]
and let us identify every codeword $w = (w_1,\ldots,w_q) \in \mathcal{C}$, with a set $w_{\text{set}} = \{(a_i,w_i) \suchthat i\in [q]\}$. Clearly, we have that $w \in A_1\times \cdots \times A_q$ if and only if $w_{\text{set}} \subseteq S$. Often, we will drop the subscript on $w_{\text{set}}$ and refer to $w$ both as a codeword, and as the set of points it passes through. Further identifying $\F_q^2$ with $[q^2]$, and  and parameterizing the problem by $r = \frac{d}{q}$, we view Andreev's Problem as a Boolean function $\text{AP}_r:\{0,1\}^{q \times q} \rightarrow \{0,1\}$.

The main challenge here is to prove (or at least conditionally disprove) $\text{NP}$-hardness for Andreev's Problem, which has been open for over $30$ years. Another natural problem one could study is the circuit complexity for $\text{AP}_r$. This is the main motivation behind this paper, and we will study the $\text{AC}^0[\oplus]$ complexity of $\text{AP}_r$. We shall eventually see that even this problem needs relatively recent results about the power of $\text{AC}^0[\oplus]$ in our proof. Informally, $\text{AC}^0$ is the class of Boolean functions computable by circuits of constant depth, and polynomial size, using $\land$, $\lor$, and $\lnot$ gates. $\text{AC}^0[\oplus]$ is the class of Boolean functions computable by circuits of constant depth, and polynomial size, using $\land$, $\lor$, $\lnot$, and $\oplus$ ($\text{MOD}_2$) gates. The interested and unfamiliar reader is referred to~\cite{AB09} (Chapter $14$) for a more formal definition and further motivation behind this class. We show that $\text{AP}_r$ cannot be computed by $\text{AC}^0$ circuits for a constant $r$. This type of result is essentially motivated by a similar trend in the study of the complexity of \emph{Minimum Circuit Size Problem}. Informally, the Minimum Size Circuit Problem (or simply $\text{MCSP}$) takes as input a truth table of a function on $m$ bits, and an integer $s$. The output is $1$ if there is a Boolean circuit that computes the function with the given truth table and has size at most $s$. It is a \emph{major} open problem to show the $\text{NP}$-hardness of $\text{MCSP}$. A lot of effort has also gone into understanding the circuit complexity of $\text{MCSP}$. Allender et al.~\cite{ABKMR06} proved a superpolynomial $\text{AC}^0$ lower bound, and Hirahara and Santanam~\cite{HS17} proved an almost-quadratic formula lower bound for $\text{MCSP}$. A recent result by Golonev et al.~\cite{GIIKKT19} extends~\cite{ABKMR06}  and proves an $\text{AC}^0[\oplus]$ lower bound for $\text{MCSP}$. Thus one can seek to answer the same question about $\text{AP}_r$.

 We now state our main theorem:

\begin{theorem}[Main Theorem]
\label{thm:AC0}
For any prime power $q$, and $r \in (0,1)$, we have that any depth $h$ circuit with $\land$, $\lor$, $\lnot$, and $\oplus$ gates that computes $\text{AP}_r$ on $q^2$ bits must have size at least $\exp\left(\tilde{\Omega}\left(hq^{\frac{c^2}{h-1}} \right) \right)$.
\end{theorem}

We make a couple of comments about the theorem. The first, most glaring aspect is that $r \in (0,1)$ is more or less a limitation of our proof technique. Of course, as $r$ gets \emph{very} small, i.e., $r = O\left(\frac{1}{q}\right)$, one can find depth $2$ circuits of size $q^{O(rq)} = q^{O(1)}$. But, we do not know that the case where, for example, $r = \Theta\left( \frac{1}{\log q}\right)$ is any easier for $\text{AC}^0$. Secondly, we are not aware of a different proof of an $\text{AC}^0$ (as opposed to an $\text{AC}^0[\oplus]$) lower bound.

\subsection{Some notation and proof ideas}

For a $p \in (0,1)$, let $X_1, X_2, \ldots$ denote independent $\Ber(p)$ random variables. For a family of Boolean functions $f:\{0,1\}^n \rightarrow \{0,1\}$, we use $f^{(n)}(p)$ to denote the random variable $f(X_1,\ldots, X_n)$. 

\begin{definition}[Sharp threshold]
For a monotone family of functions $f$, we say that $f$ has a \emph{sharp threshold} at $p$ if for every $\epsilon > 0$, there is an $n_0$ such that for every $n > n_0$, we have that $\Prob(f^{(n)}(p(1 - \epsilon)) = 0) \geq 0.99$, and $\Prob(f^{(n)}(p(1 + \epsilon)) = 1) \geq 0.99$.
\end{definition}

Henceforth, we shall assume that $q$ is a very large prime power. So, all the high probability events and asymptotics are as $q$ grows. Where there is no ambiguity, we also just use $f(p)$ to mean $f^{(n)}(p)$ and $n$ growing.

One limitation of $\text{AC}^0[\oplus]$ that is exploited when proving lower bounds (including in~\cite{GIIKKT19}) for monotone functions is that $\text{AC}^0[\oplus]$ cannot compute functions with `very' sharp thresholds. For a quantitative discussion, let us call the smallest $\epsilon$ in the definition above the \emph{threshold interval}. It is known that $\text{AC}^0$ (and therefore, $\text{AC}^0[\oplus]$) can compute (some) functions with threshold interval of $O\left(\frac{1}{\log n}\right)$, for example, consider the following function on Boolean inputs $z_1,\ldots,z_n$: Let $Z_1 \sqcup \cdots \sqcup Z_{\ell}$ be an equipartition of $[n]$, such that each $|Z_i| \approx \log n  - \log \log n $. Consider the function given by 

\[
f(z_1,\ldots, z_n) = \bigvee_{i \in [\ell]} \left( \bigwedge_{j \in Z_i} z_j \right).
\]

This is commonly known as the \emph{tribes} function and is known to have a threshold interval of $O\left( \frac{1}{\log n}\right)$. This is clearly computable by an $\text{AC}^0$ circuit. A construction from~\cite{LSSTV19} gives an $\text{AC}^0$ circuit (in $n$ inputs) of size $n^{O(h)}$ and depth $h$ that has a threshold interval $\tilde{O}\left( \frac{1}{(\log n)^{h - 1}}\right)$. A remarkable result from~\cite{LSSTV19} and~\cite{CHLT19} (Theorem 13 from~\cite{CHLT19} and Lemma $3.2$ in~\cite{A19}) says that this is in some sense, tight. Formally,

\begin{theorem}[~\cite{LSSTV19}~\cite{CHLT19}]
\label{thm:LSSTV}
Let $n$ be any integer and $f :\{0,1\}^n \rightarrow \{0,1\}$ be a function with threshold interval $\delta$ at $\frac{1}{2}$. Any depth $h$ circuit with $\land$, $\lor$, $\lnot$, and $\oplus$ gates that computes $f$ must have size at least $\exp\left(\Omega\left(h \left(1/\delta\right)^{\frac{1}{h - 1}} \right)\right)$.
\end{theorem}

In~\cite{LSSTV19}, this was studied as the \emph{Coin Problem}, which we will also define in Section~\ref{sec:AC0}. Given the above theorem, a natural strategy suggests itself. If we could execute the following two steps, then we would be done:

\begin{itemize}
\item[1.] Establish Theorem~\ref{thm:LSSTV} for functions with thresholds at points other than $\frac{1}{2}$.
\item[2.] Show that $\text{AP}_r$ has a sharp threshold at $q^{-r}$ with a suitably small threshold interval, i.e., $\frac{1}{\operatorname{poly} q}$.\end{itemize}

The first fact essentially reduces to approximating $p$-biased coins by unbiased coins in constant depth. Though we are unable to find a reference for this, this is relatively straightforward, and is postponed to Appendix~\ref{sec:asec4}. Understanding the second part, naturally leads us to study $\text{AP}_r(p)$ for some $p = p(q)$. Let $A_1,\ldots, A_q \subset \F_q$ be independently chosen random subsets where each element is included in $A_i$ with probability $p$. Let $\mathcal{C}$ be the $\RS[q,rq]$ code. We have $|\mathcal{C}| = q^{rq + 1}$. Let us denote 

\[
X := |(A_1\times \cdots \times A_q) \cap \mathcal{C}|.
\] 

For $w \in \mathcal{C}$, let $X_w$ denote the indicator random variable for the event $\{w \in A_1 \times \cdots \times A_q\}$. Clearly, $X = \sum_{w \in \mathcal{C}}X_w$, and for every $w \in \mathcal{C}$, we have $\Prob(X_w = 1) = p^{q}$. We first note that for $\epsilon = \omega\left(\frac{\log q}{q}\right)$, and $p = q^{-r}(1 - \epsilon)$, we have, using linearity of expectation,

\begin{align*}
\E[X]  & = \sum_{w \in \mathcal{C}}\E[X_w] \\
& = |\mathcal{C}| \cdot (q^{-r}(1 - \epsilon))^q \\
& = q^{rq + 1}\left(q^{-r}(1 - \epsilon\right))^q \\
& = q\cdot (1 - \epsilon)^q \\
& \leq q \cdot e^{-\epsilon q} \\
& = o(1).
\end{align*}

When $p = q^{-r}(1 + \epsilon)$, using a similar calculation as above, we have

\[
\E[X] = q \cdot (1 + \epsilon)^q \geq q.
\] 

To summarize, for $\epsilon = \omega \left( \frac{\log q}{q}\right)$, and $p = q^{-r}(1 - \epsilon)$, $\E[X] \rightarrow 0$, and for $p = q^{-r}(1 + \epsilon)$, $\E[X] \rightarrow \infty$.

\begin{lemma}
\label{lem:smallp}
For $\epsilon =  \omega\left(\frac{\log q}{q}\right)$, we have 
\[
\Prob(\text{AP}_r(q^{-r}(1 - \epsilon)) = 1) \leq \exp\left(- \Omega(\epsilon q)\right).
\]
\end{lemma}

\begin{proof}
This is just Markov's inequality. We have $\Prob(\text{AP}_r(p(1- \epsilon)) = 0) = \Prob(X \geq 1) \leq \E[X] \leq q \cdot e^{-\epsilon q} = \exp\left(-\Omega(\epsilon q)\right)$.
\end{proof}

This counts for half the proof of the sharp threshold for $\text{AP}_r$. The other half forms the main technical contribution of this work. We show the following:

\begin{theorem}
\label{thm:sharpthreshold}
Let $q$ be a prime power, $r = r(q)$ and $\epsilon = \epsilon(q)$ be real numbers such that $q^{-r} \geq \frac{\log q}{q}$ and $\epsilon = \omega\left(\max\left\{q^{-r}, \sqrt{q^{r - 1}  \log \left(q^{1 - r}\right)}\right\}\right)$.

Let $A_1,\ldots,A_q$ be independently chosen random subsets of $\F_q$ with each point picked independently with probability $q^{-r}(1 + \epsilon)$. Then 
\[
\Prob((A_1 \times \cdots \times A_q) \cap \RS[q,rq] = \emptyset) = o(1).
\]
\end{theorem}

There is a technical condition on $\epsilon$ that can be ignored for now, and will be addressed before the proof. The only relevant thing to observe is that when $r$ is bounded away from $0$ and $1$, then $\epsilon = \frac{1}{\operatorname{poly} (q)}$ suffices.  The condition to focus on here is that $q^{-r} \geq \frac{\log q}{q}$. Indeed, one can see that this condition is necessary to ensure that w.h.p, all the $A_i$'s are nonempty. So, for example, if the dimension of $\mathcal{C}$ is $q - 1$, then setting $p = q^{-1}(1 + \epsilon)$ is enough for $\E[X] = \omega(1)$ but this does not translate to there almost surely being a codeword in $A_1 \times \cdots \times A_q$.

Lemma~\ref{lem:smallp} and Theorem~\ref{thm:sharpthreshold} together give us that $\text{AP}_r$ has a sharp threshold at $\max \left\{ q^{-r}, \frac{\log q}{q} \right\}$ whenever $1 - \frac{1}{q} \geq r \gg \frac{1}{\log p}$. For the sake of completeness one could ask if $\text{AP}_r$ has a threshold for all feasible values of $r$, and we show that the answer is yes. More formally,
 
 \begin{theorem}[Sharp threshold for list recovery]
 \label{thm:fullrange}
 For every $r = r(q)$, there is a critical $p = p(r,q)$  such that for every $\epsilon >0$,
 \begin{enumerate}
 \item $\Prob\left(\text{AP}_{r}(p(1 - \epsilon)) = 1\right) = o(1)$.
 \item $\Prob\left(\text{AP}_{r}(p(1 + \epsilon)) = 1\right) = 1 - o(1)$.
 \end{enumerate}
 \end{theorem}

The case that is not handled by Theorem~\ref{thm:sharpthreshold} is when $r = O\left(\frac{1}{\log q} \right)$ (since in this case, Theorem~\ref{thm:sharpthreshold} requires $\epsilon = \Omega(1)$). This corresponds to the case where $q^{-r}$ is a number bounded away from $0$ and $1$.

\subsection{What doesn't work, and why}

One obvious attempt to prove Theorem~\ref{thm:sharpthreshold} is to consider the second moment of $X (= |\mathcal{C} \cap (A_1\times\cdots\times A_q)|)$ and hope that $\E[X^2] = (1 + o(1))\E^2[X]$. Unfortunately, $\E[X^2]$ is too large. Through a very careful calculation using the weight distribution of Reed Solomon codes which we do not attempt to reproduce here, we have $\E[X^2] = \Omega\left( e^{\frac{1}{p}}\E^2[X]\right)$. So in the regime where, for example, $p = q^{-\Omega(1)}$, this approach is unviable.

To understand this (without the aforementioned involved calculation) in an informal manner, let us fix $p = q^{-r}$ for some fixed constant $r$. Let us identify the tuple of sets $(A_1,\ldots, A_q)$ with the single set $S = \cup_{i \in [q]}\{(a_i,z) \suchthat z \in A_i\}$. So, we are choosing a random subset $S \subset \F_q^2$ of size $\approx q^{2-r}$. On the other hand, the objects we are looking for, i.e., codewords, have size $q$. This is much larger than the standard deviation of $|S|$, which is of the order of $q^{1 - (r/2)}$. Thus, conditioning on the existence of some codeword $w \subset \F_q^2$, the distribution of $S$ changes significantly. One way to see this is the following: Using standard Chernoff bounds, one can check that the size of $S$ is almost surely $q^{2 - r} \pm O\left(q^{1 - (r/2)} \log q \right)$. However, conditioned on $w \in A_1 \times \cdots \times A_q$, the size of $S$ is almost surely $q + q^{-r}(q^2 - q) \pm O\left(q^{1 - (r/2)}\log q\right)$ (the additional $q$ comes from the points that make up $w$). This is much larger than before when $r$ is relatively large. On the other hand, the main point behind (successful) applications of the second moment method is that the distribution does not significantly change after such a conditioning.

One possible way to circumvent the above problem is to pick a uniformly random set $S \subset \F_q^2$ of size $q^{2 - r}$, instead of every point independently with probability $q^{-r}$. This is a closely related distribution, and it is often the case that Theorems in this model are also true in the above `i.i.d.' model. This fact can be also be made formal (see, for example~\cite{JLR00} Corollary $1.16$). Here, when one conditions on the existence of some codeword $w$, at least $|S|$ does not change. Thus the second moment method is not ruled out right at the start. However, it seems to be much more technically involved and it is unclear if it is possible to obtain the relatively small threshold interval that is required for Theorem~\ref{thm:AC0} in this way.

\subsection{What works and how}

Here, we sketch the proofs of the Theorem~\ref{thm:sharpthreshold} and Theorem~\ref{thm:fullrange}, which can be considered the two main technical contributions of this work.

\subsubsection{Proof sketch of Theorem~\ref{thm:sharpthreshold}} 

The key idea in the proof of this theorem is to count the number of polynomials in the `Fourier basis'. Let us consider $f :\F_q^q \rightarrow \{0,1\}$ to be the indicator of $\mathcal{C}$. For $i \in [q]$, let $g_i : \F_q \rightarrow \{0,1\}$ denote the indicator of $A_i$. 

For an extremely brief and informal discussion, what we what we want is essentially $\langle f, \prod_{i \in [q]}g_i \rangle$, which, by Plancharel's identity (see Fact~\ref{fact:Plancherel}) is $\sum_{\alpha}\widehat{f} \cdot  \widehat{\prod_ig_i}(\alpha)$. Since $\mathcal{C}$ is a vector space, we have that $\widehat{f}$ is supported on $\mathcal{C}^{\perp}$. Moreover, $\widehat{g_i}(\alpha_i)$ is much larger when $\alpha_i = 0$ than when $\alpha_i \neq 0$ if $A_i$ is random. This combined with the fact that most points in $\mathcal{C}^{\perp}$ have large weight, and a bit more Fourier analysis means that the inner product, $\langle f,\prod_i g_i\rangle$ is dominated by $\widehat{f}(0) \prod_{i \in [q]}\widehat{g_i}(0)$ which is the expected number of codewords in $A_1 \times \cdots \times A_q$.

Now we give a slightly less informal overview. What we are trying to estimate is exactly

\begin{align*}
X = |\mathcal{C} \cap (A_1\times \cdots \times A_q)| =  \sum_{(x_1,\ldots,x_q) \in \F^q}f(x)\left( \prod_{i \in [q]}g_i(x_i)\right).
\end{align*}

Using Fourier analysis over $\F_q$, one can show that

\begin{align*}
\frac{q^q}{|\mathcal{C}|} \cdot X & = \sum_{(\alpha_1,\ldots,\alpha_q) \in \mathcal{C}^{\perp}}\prod_{i \in [q]}\widehat{g_i}(\alpha_i) \\
& \geq \prod_{i \in [q]}\widehat{g_i}(0) - \left|\sum_{(\alpha_1,\ldots, \alpha_q) \in \mathcal{C}^{\perp}}\left(\prod_{i \in [q]}\widehat{g_i}(\alpha_i)\right)\right|.
\end{align*}

Using the fact that $\mathcal{C}$ is an $\RS[q,rq]$ code, one has (see Fact~\ref{fact:dual}) that $\mathcal{C}^{\perp}$ is an $\RS[q, q - rq - 1]$ code. What will eventually help in the proof is that the weight distribution of Reed Solomon codes (and so in particular, $\mathcal{C}^{\perp}$) is well understood (see Theorem~\ref{thm:MDSwtdist}). 

Now clearly, it suffices to understand the term $\sum_{(\alpha_1,\ldots, \alpha_q) \in \mathcal{C}^{\perp}}\left(\prod_{i \in [q]}\widehat{g_i}(\alpha_i)\right) =: R$. One way to control $|R|$ is to control $|R|^2 = R \overline{R}$. Here, one can use the fact that the $A_i$'s are randomly and independently chosen to establish cancellation in many terms of $\E[|R|^2]$. More formally, one can prove that

\[
\E[|R|^2] = \sum_{(\alpha_1,\ldots,\alpha_q) \in \mathcal{C}^{\perp} \setminus \{\overline{0}\}}\prod_{i \in [q]}\E[|\widehat{g_i}(\alpha_i)|^2].
\]

It is a more or less standard fact that if $A_i$ is a uniformly random set of size $pq =: t$, then
\[
\E[|\widehat{g_i}(0)|^2] \sim \left( \frac{t}{q}\right)^2
\]

and

\[
\E[|\widehat{g_i}(\alpha_i)|^2] \sim \frac{t}{q^2}
\]
 
 for $\alpha_i \neq 0$. This difference, will be the reason why $|R|$ is typically much smaller than $\prod_{i \in [q]}\widehat{g_i}(0)$. To continue, let us believe the heuristic that most polynomials over $\F_q$ of degree $\Theta(q)$ have very few ($o(q)$) zeroes, we can use the rough estimate:
 \begin{align*}
 \E[|R|^2] & \approx |\mathcal{C}^{\perp}|\left( \frac{t}{q^2}\right)^{q - o(q)} \\
 & \approx q^{q - rq}\left(\frac{p}{q}\right)^{q - o(q)} \\
 & \approx q^{-rq}p^{q - o(q)}.
 \end{align*}
 
 And so, Markov's Inequality gives that $|R|$ is unlikely to be much greater than $q^{\frac{rq}{2}}p^{\frac{q}{2} + o(q)}$. On the other hand, with high probability,  
 \begin{align*}
 \prod_{i \in [q]} \widehat{g_i}(0) & \approx \left(\frac{t}{q}\right)^q \\
 & \approx p^{q}.
 \end{align*}
 
Thus if $p \geq q^{-r + o(1)}$, we have that $(q^q / |\mathcal{C}|)\cdot X \geq \prod_{i \in [q]}\widehat{g_i}(0) - |R| > 0$, and so in particular, $X > 0$. The proof of Theorem~\ref{thm:sharpthreshold} is essentially a much tighter, and more formal version of the above argument, and is postponed to Appendix~\ref{sec:sharpthreshold}.

\subsubsection{Proof sketch of Theorem~\ref{thm:fullrange}}

The starting point of Thoerem~\ref{thm:fullrange} is noticing that the only case not covered by Theorem~\ref{thm:sharpthreshold} is $p \in (0,1)$ is some fixed constant, or equivalently $r = O\left( \frac{1}{\log q}\right)$. Here we have a  somewhat crude weight distribution result for Reed Solomon codes (Proposition~\ref{prop:RSwtdist}) to compute the second moment. We first show that $\E[X^2] = O\left( e^{\frac{1}{p}} \E^2[X]\right)$. Using, for example the Paley-Zygmund Inequality~(\ref{eqn:PaleyZygmund}), this means that $\Prob(X > 0) \geq \Omega(e^{-\frac{1}{p}})$. Thus we have that $\{X > 0\}$ with at least some (possibly small) constant probability. But what we need is that $\Prob(X > 0) \geq 0.99$. For this, we now use the fact that $\text{AP}_r$ is monotone, and transitive-symmetric, which informally means that any two variables of $\text{AP}_r$ look the same (see Definition~\ref{defn:transitive} for a formal definition). Standard applications of hypercontractivity for the Boolean hypercube (see Theorem~\ref{thm:FK}) gives that for $p' = p + O\left( \frac{1}{\log q}\right)$, we have that $\Prob(\text{AP}_r(p') = 1) \geq 0.99$.

The details of this proof are postponed to Appendix~\ref{sec:fullrange}.

One thing to note is that our definition of sharp threshold only makes sense when the critical probability $p_r$ is bounded away from $1$ (since otherwise trivially there is some function $\epsilon = \epsilon(q) = o(1)$ such that $p\cdot (1 + \epsilon) = 1$). So, we will restrict ourselves to the regime where $r = \Omega\left( \frac{1}{\log q}\right)$. Also, it is to be understood that all the statements above (and below) only make sense when $rq$ is an integer, and thus we shall restrict ourselves to this case.

Finally, we address the question of random list recovery with errors as another application of Theorem~\ref{thm:sharpthreshold}.

\subsection{Random list recovery with errors}

Given a random subset of points in $S \subseteq \F_q^2$, what is the largest fraction of any degree $d = \Theta(q)$ polynomial that is contained in this set? Using the Union Bound, it is easy to see that no polynomial of degree $d$ has more than $d \log_{\frac{1}{p}}q + o(q)$ points contained in $S$ (formal details are given in Section~\ref{sec:expagreement}). We show that perhaps unsurprisingly, this is the truth. Formally,

\begin{corollary}
\label{corr:errors}
Let $S$ be a randomly chosen subset of $\F_q^2$ where each point is picked independently with probability $p$. Then with probability $1 - o(1)$, 
\[
 \max_{w \in \RS[q,d]}|w \cap S| = d \log_{\frac{1}{p}} q  - O \left( \frac{q}{\log \left( \frac{1}{p}\right)}\right).
 \]
\end{corollary}

 We restrict our attention to the case when $d = \Theta(q)$, where the above statement is nontrivial. This is the content of Section~\ref{sec:expagreement}. However, we believe that the statement should hold for all rates, and error (in general) better than $O\left( \frac{q}{\log q}\right)$.
 
 We make two final comments before proceeding to the proofs: (1) In Theorem~\ref{thm:sharpthreshold}, each $A_i$ is chosen by including each point independently. However, the same proof works if $A_i$ is a uniformly random set with a prescribed size. (2) Although we only state the lower bound for $\text{AC}^0[\oplus]$, one can check that all the tools (and, therefore, the lower bound) still work when we replace the $\oplus$ gates with any $\oplus_p$ ($\text{MOD}_p$) for any small prime $p$.

\section{$\text{AC}^0[\oplus]$ lower bound for $\text{AP}_r$}
\label{sec:AC0}

We prove the lower bound by showing that $\text{AP}_r$ solves a biased version of the \emph{Coin Problem}, and use the lower bounds known for such kinds of functions, obtained by~\cite{LSSTV19},~\cite{CHLT19}. 

\begin{definition}[$(p,\epsilon)$-coin problem]
We say that a circuit $C = C^{n}$ on $n$ inputs solves the $(p,\epsilon)$-coin problem if
\begin{itemize}
\item For $X_1,\ldots, X_n \sim \operatorname{Ber}(p(1 - \epsilon))$,
\[\Prob(\text{C}(X_1,\ldots,X_n) = 0) \geq 0.99 \]
\item For $X_1,\ldots, X_n \sim \operatorname{Ber}(p(1 + \epsilon))$,
\[\Prob(\text{C}(X_1,\ldots,X_n) = 1) \geq 0.99 \]
\end{itemize}
\end{definition}

We shall abbreviate the $(p,\epsilon)$-coin problem on $n$ variables as $\text{CP}^n(p,\epsilon)$. We observe that a function $f:\{0,1\}^n \rightarrow \{0,1\}$ \emph{solves} $\text{CP}^n\left(p,\epsilon \right)$ if it has a sharp threshold at $p$ with threshold interval at most $\epsilon$. The one obstacle we have to overcome in using Theorem~\ref{thm:LSSTV} is that $\text{AP}_r$ has a sharp threshold at $p^{-c} \ll \frac{1}{2}$. However, we will show how to simulate biased Bernoulli r.v's from almost unbiased ones. Let $\text{C}(s,d)$ to denote the class of functions on $n$ variables which have circuits of size $O(s) = O(s(n))$ and depth $d = d(n)$ using $\land$, $\lor$, $\lnot$, and $\oplus$ gates. Here, we make the following simple observation about the power of $\text{AC}^0[\oplus]$ circuits to solve biased and unbiased $\epsilon$-coin problem. First, we observe that it is possible to simulate a biased coin using an unbiased one. 

\begin{lemma}
\label{lem:bias}
Let $s$ be such that $\frac{1}{2^s} \leq p \in (0,1)$, and $\epsilon \leq \frac{1}{s^K}$ for a large constant $K$. Then, there is a CNF $F_p$ on $t \leq s^2$-variables such that for inputs $X_1 \ldots, X_t \in \operatorname{Ber}\left(\frac{1}{2} + \epsilon \right)$,
\[
\Prob\left(F_p(X_1,\ldots, X_t) = 1\right) = p(1 + \Omega(\epsilon L))
\]
and for inputs $X_1 \ldots, X_t \in \operatorname{Ber}\left(\frac{1}{2} - \epsilon \right)$,
\[
\Prob\left(F_p(X_1,\ldots, X_t) = 1\right) = p\left(1 + \frac{1}{2^{\Omega(\sqrt{t})}} - \Omega(\epsilon L)\right)
\]
where $L = \lfloor\log_2(1/p)\rfloor$.
\end{lemma}

The idea is essentially that the $\texttt{AND}$ of $k$ unbiased coin is a $2^{-k}$-biased coin. However, some extra work has to be done if we want other biases (say, $(0.15) \cdot 2^k$). The proof of this lemma is postponed to the Appendix~\ref{sec:asec4}. This lemma now gives us the following:

\begin{lemma}
\label{lem:biasedcoin}
Let $z \in (0,1)$ be a fixed constant. If $\text{CP}^n\left(\frac{1}{n^{z}},o(\epsilon \log n)\right) \in \text{C}^n(s,h)$, then there is a $t \leq \log^2n$ such that $\text{CP}^{nt}\left(\frac{1}{2}, \epsilon \right) \in \text{C}^{nt}(z s\log n, h+2)$.
\end{lemma}

\begin{proof}
Let $C$ be a circuit for $\text{CP}^n\left(\frac{1}{n^{z}},\delta\right)$-coin problem. Replace each input variable with the CNF $F_{\left(\frac{1}{n^{z}}\right)}$ from Lemma~\ref{lem:bias} on $t = O(\log^4 n) $ independent variables. Call this circuit $\text{C}'$, on $tn$ variables. If the bias of each of these input variables is $\frac{1}{2} + \epsilon$, then the guarantee of Lemma~\ref{lem:bias} is that output of the and gate is $1$ with probability at least $\frac{1}{n^{z}}(1 + \Omega(\epsilon \log n))$. A similar computation gives that if the bias of the inputs are $\left( \frac{1}{2} - \epsilon \right)$, then the bias of the output is at most $\frac{1}{n^{z}}(1 -  \Omega(\epsilon \log n))$. Therefore, $\text{C}'$ solves $\text{CP}^{nt}\left(\frac{1}{2}, \epsilon\right)$, and has size at most $s \log n$, and depth $h+2$.

\end{proof}

Theorem~\ref{thm:sharpthreshold} and Lemma~\ref{lem:smallp}, together, now give us the following corollary:

\begin{corollary}
Let $q$ be a large enough prime power. Then $AP_r$ on $q^2$ inputs solves the $\left(q^{-r}, \epsilon\right)$ coin problem, for $\epsilon = \omega\left(\max\left\{q^{-r}, q^{\frac{r - 1}{3}}\right\}\right)$
\end{corollary}

As a result, Theorem~\ref{thm:LSSTV}, and Lemma~\ref{lem:bias}, and Lemma~\ref{lem:biasedcoin} together, give us the following bounded depth circuit lower bound for $\text{AP}_r$:

\begin{theorem*}[\ref{thm:AC0}, Restated]
For any $r \in (0,1)$, and $h \in \N$, we have that 
\[
\text{AP}_{r} \not \in  \text{C}\left(\exp\left\{\tilde{\Omega}\left(hq^{\frac{r^2}{h-1}} \right) \right\},h\right).
\]
\end{theorem*}

\section{Random list recovery with errors}
\label{sec:expagreement}

In this section, we shall again consider Reed-Solomon codes $\RS[q,rq]$ where $r$ is some constant between $0$ and $1$. Let us slightly abuse notation, as before, and think of a codeword $w \in \RS[q,rq]$ corresponding to a polynomial $p(X)$ as the set of all the zeroes of the polynomial $Y = p(X)$. That is, for a codeword $w = (w_1,\ldots, w_q)$ associated with polynomial $p$, we think of $w$ as a subset $\{(a_i,p(a_i)) \suchthat i \in [q]\}$ (recall that $\F = \{a_1,\ldots, a_q\}$). For a set of points $S \subset \F_q^2$ and a codeword $w$ we say the \emph{agreement between $w$ and $S$} to denote the quantity $|w \cap S|$. For a code $\mathcal{C}$, we say that the \emph{agreement between $\mathcal{C}$ and S} to denote $\max_{w \in \mathcal{C}}|m \cap S|$.

We are interested in the following question: For a set $S \subset \F_q^2$. What is the smallest $\ell$ such that there exists a $w \in \RS[q,rq]$ such that $|w \cap S| \geq q - \ell$? In other words, what is the largest agreement between $\RS[q,rq]$ and $S$? This is (very close to) the list recovery problem for codes \emph{with errors}. Naturally, we seek to answer this question when $S$ is chosen randomly in an i.i.d. fashion with probability $p$. Theorem~\ref{thm:sharpthreshold} gives asymptotically tight bounds in a relatively straightforward way for constant error rate. 

One can observe that the only properties about Reed-Solomon codes that was used in Theorem~\ref{thm:sharpthreshold} was the weight distribution in the dual space of codewords. However, (see Appendix~\ref{sec:RS}) these are also true for punctured Reed-Solomon codes codes. So, an analogus theorem also holds for punctured Reed Solomon codes. Formally,

\begin{theorem}
Let $q,n,d$ be integers such that $q$ is a prime power and $n = \omega(\log q)$, and $q^{-\frac{d}{n}} \geq \frac{\log n}{q}$ and let $\epsilon = \omega \left(\max\left\{q^{- \frac{d}{n}} \sqrt{q^{1 - \frac{d}{n}}\log \left(q^{1 - \frac{d}{n}}\right)}\right\}\right)$. Let $\mathcal{C}$ be an $\RS[q,d]|_n$ code.

Let $A_1,\ldots,A_n$ be independently chosen random subsets of $\F_q$ with each point picked independently with probability $q^{-\frac{d}{n}}(1 + \epsilon)$. Then 
\[
\Prob((A_1 \times \cdots \times A_n) \cap \mathcal{C} = \emptyset) = o(1).
\]
\end{theorem}

We do not repeat the proof but is it the exact same as that of Theorem~\ref{thm:sharpthreshold}. Let $\mathcal{E}_a$ denote the event that the agreement between $S$ and $\RS[q,rq]$ is $a$. Union bound gives us that

\begin{equation}
\label{eqn:expagreement}
\Prob(\mathcal{E}_{q - \ell}) \leq \binom{q}{\ell}q^{rq+1}p^{q - \ell }.
\end{equation}
\
So if $\ell$ is such that the RHS of~\ref{eqn:expagreement} is $o(1)$. Then the agreement is almost surely less than $q - \ell$. For the other direction, we have the following corollary:

\begin{corollary}
\label{corr:punctured}
Let $\epsilon \geq \max\left\{10q^{- \frac{d}{q - \ell}}, \sqrt{q^{1 - \frac{d}{q - \ell}} \cdot \log q }\right\}$. Let $S$ be a randomly chosen subset of $\F_q^2$ with each point picked independently with probability at least $q^{-\frac{d}{q - \ell}}(1 + \epsilon)$, then with probability at least $1 - o(1)$, the agreement between $S$ and $\RS[q,d]$ is at least $q - \ell $.
\end{corollary}

\begin{proof}
For $i \in [q - \ell]$, let us denote 
\[
S_i:= \{j \suchthat (i,j) \in S\}.
\]
Let us use $S' := S_1\times\cdots \times S_{q - \ell }$. Let us denote $\mathcal{C} = \RS[q,d]|_{q - \ell }$.  Formally, for a codeword $w \in \RS[q,d]$, denote $p_w$ to be the polynomial corresponding to $m$. We have
\[
\mathcal{C} = \{(i, p_w(i)) \suchthat i \in [q - \ell])\}
\]
We observe that the conditions in Theorem~\ref{thm:sharpthreshold} hold, so
\[
\Prob(C' \cap S' = \emptyset) = o(1)
\]
as desired.
\end{proof}

\begin{corollary*}[\ref{corr:errors}, restated]
Given a random subset $S \subseteq \F_q^2$ where each point is picked with probability $p$, then with probability at least $1 - o(1)$, the largest agreement $\RS[q,d]$ with $S$ is $ d \log_{\frac{1}{p}} q  - O \left( \frac{q}{\log \left( \frac{1}{p}\right)}\right)$.
\end{corollary*}

\begin{proof}
Let $a$ be an integer that denotes the maximum agreement between $S$ and $\RS(q,d)$. Suppose that $a \leq  d \log_{\frac{1}{p}} q$, then  setting $\ell = q - a$, and noting that the conditions for Corollary~\ref{corr:punctured} are satisfied, we get that with probability at least $1 - o(1)$, there is a polynomial that agrees with the set $S$ in the first $q - \ell$ coordinates. On the other hand, if $a \geq d\log_{ \frac{1}{p}}q + 4\frac{q}{\log \left( \frac{1}{p}\right)}$, again, setting $\ell = q - a$, Union Bound gives us:

\begin{align*}
\Prob(\mathcal{E}_{q - \ell}) & \leq \sum_{w \in \mathcal{C}} \sum_{\substack{P \subset \F_q \\ |P| = q - \ell}} \Prob(w|_P \subseteq S)\\
& = \binom{q}{\ell}q^{d+1}p^{q - r} \\
& \leq \left( e\frac{q}{\ell}\right)^{\ell} q^{d+1}p^{q - \ell} \\
& \leq e^qq^{d+1}p^{q - \ell} \\
& \ll \frac{1}{q} .
\end{align*}

And so we have that with probability at least $1 - o(1)$, the agreement of $\RS(q,d)$ with $S$ is $d \log_{\frac{1}{p}} q  - O\left( \frac{q}{\log \left( \frac{1}{p}\right)}\right)$.
\end{proof}

\section{Conclusion}

We started off by attempting to prove a bounded depth circuit lower bound for Andreev's Problem. This led us into (the decision version of the) random List Recovery of Reed-Solomon codes. Here we show a sharp threshold for a wide range of parameters, with nontrivial threshold intervals in some cases. However, one of the unsatisfactory aspects about Theorem~\ref{thm:fullrange} is that it is proved in a relatively `hands-off' way possibly resulting in a suboptimal guarantee on $\epsilon$. The obvious open problem that is the following:

\paragraph{Open Problem:} Is Theorem~\ref{thm:fullrange} with a better bound on $\epsilon$?

If it is true with a much smaller $\epsilon$, it would extend in a straightforward way to the $\text{AC}^0[\oplus]$ lower bound as well. Another point we would like to make is that the only thing stopping us from proving Theorem~\ref{thm:fullrange} for general $\text{MDS}$ codes is the lack of Proposition~\ref{prop:symm}

\paragraph{Acknowledgements} I am extremely grateful to Amey Bhangale, Suryateja Gavva, and Mary Wootters for the helpful discussions. I am especially grateful to Nutan Limaye and Avishay Tal for explaining~\cite{LSSTV19} and~\cite{CHLT19} respectively to me. I am grateful to Partha Mukhopadhyay for suggesting the $\text{AC}^0[\oplus]$ lower bound problem for list recovery. I am also extremely grateful to Swastik Kopparty for the discussions that led to Corollary~\ref{corr:errors}, and to Bhargav Narayanan for the discussions that led to Theorem~\ref{thm:fullrange}.

\bibliographystyle{alpha}
\bibliography{references.bib}

\appendix

\section{More preliminaries}
\label{sec:prelims}

\subsection{Properties of Reed-Solomon codes}
\label{sec:RS}

The first fact we will use is that the dual vector space of a Reed-Solomon code is also a Reed-Solomon code.

\begin{fact}
\label{fact:dual}
Let $\mathcal{C} := \RS[q,d]$. Then $\mathcal{C}^{\perp} = \RS[q,q - d - 1]$.
\end{fact}

For $t \neq 0$, let $W_t$ be the number of codewords of weight $t$ in $\RS[q,d]$. This is a relatively well understood quantity.

\begin{theorem}[\cite{EGS09}]
\label{thm:MDSwtdist}
We have:
\[
|W_{q-i}| = \binom{q}{i}\sum_{j = 0}^{d - i }(-1)^j\binom{q-i}{j}(q^{d - i - j +1} - 1).
\]
\end{theorem}

However, we just need the following slightly weaker bound that is easier to prove:

\begin{proposition}
\label{prop:RSwtdist}
We have $W_{q-i} \leq \frac{q^{d+1}}{i!} $.
\end{proposition}

\begin{proof}
We have that $\mathcal{C}$ is a $d+1$-dimensional subspace of $\F^q$. Add $i$ extra constraints by choosing some set of $i$ coordinates and restricting them to $0$. As long as $i < d$, these new constraints strictly reduce the dimension of $\mathcal{C}$. There are exactly $\binom{q}{i}$ ways to choose the coordinates, and the resulting space has dimension $d+1 - i$. 
Therefore, the number of codewords of weight at most $q - i$ is at most $q^{d+1 -i} \cdot \binom{q}{i}\cdot  \leq \frac{q^{d+1}}{i!}$.
\end{proof}
The above bound is asymptotically only a factor of $e$ away for small values of $i$.

\subsubsection{Punctured Reed-Solomon codes}

All of the statements above when instead of Reed-Solomon codes, one considers \emph{punctured} Reed-Solomon codes. For a $w = (w_1,\ldots,w_n) \in \F_q^n$, and a set $S \subset [n']$, let us define
\[
w|_S = (w_i)_{i \in S}.
\]
For a subset $\mathcal{C} \subset \F_q^{n'}$, let us define
\[
\mathcal{C}|_S : = \{w|_S\suchthat w \in \mathcal{C}\}
\]

We call $\RS[q,d]|_S$ the $S$-punctured $\RS[q,d]$ code. Let $\mathcal{C}$ denote the $\RS[q,d]|_n$ code. Since the properties we will care about are independent of the specific set $S$, let is just parametrize this by $|S| =: n$. The following properties hold

\begin{enumerate}
\item $\mathcal{C}^{\perp} = \RS[q, q - d - 1]_n$.
\item Let $W_{i}$ be the number of codewords in $\mathcal{C}$ code of weight $i$. Then we have 
\[
W_{n-i} \leq \frac{q^{k - i}n^i}{i!} .
\]
\end{enumerate}

Both facts can be easily checked.

\subsection{Basic probability inequalities}
\label{sec:probability}

We will use the standard (multiplicative) Chernoff bound for sums of i.i.d. Bernoulli random variables. Let $X_1,\ldots,X_n$ be independent $\operatorname{Ber}(p)$ random variables. Let $X := \sum_{i \in [n]}X_i$and denote $\mu = \E[X] = np$. Then for any $\epsilon \in (0,1)$, we have:

\begin{equation}
\label{eqn:Chernoff}
\Prob\left(|X - \mu|\geq \epsilon\mu \right) \leq e^{\frac{\epsilon^2\mu}{2}}.
\end{equation}

We also have (a special case of) the Paley-Zygmund inequality, which states that for a nonnegative random variable $X$, 
\begin{equation}
\label{eqn:PaleyZygmund}
\Prob(X > 0)\geq \frac{\E^2[X]}{\E[X^2]}.
\end{equation}

\subsection{Fourier analysis over $\F_q$}
\label{sec:Fourier}

For functions $u,v: \F_q^n \rightarrow \C$, we have a normalized inner product $\langle u,v \rangle := \frac{1}{q^n}\sum_{s \in \F_q^n}u(s)\overline{v(s)}$. Consider any symmetric, non-degenerate bi-linear map $\chi:\F_q^n \times \F_q^n \rightarrow \R/\Z$ (such a map exists). For an $\alpha \in \F_q^n$, the \emph{character} function associated with $\alpha$, denoted by $\chi_{\alpha}:\F_q^n \rightarrow \C$ is given by $\chi_{\alpha}(x) = e^{-2\pi i \chi(\alpha, x)}$.
 
 We have that for all distinct $\alpha, \beta \in \F_q$, we have that $\langle\chi_{\alpha},\chi_{\beta}\rangle = 0$, and every function $f :\F_q \rightarrow \C$ can be written in a unique way as $f(x) = \sum_{\alpha \in \F_q}\widehat{f}(\alpha)\chi_{\alpha}(x)$. Here the $\widehat{f}(\alpha)$'s are called the \emph{Fourier coefficients}, given by 
 \[
 \widehat{f}(\alpha) = \langle f,\chi_{\alpha} \rangle.
 \]
 
 We will state some facts that we will use in the proof of Theorem~\ref{thm:sharpthreshold}.  The interested reader is referred to the excellent book of Tao and Vu~\cite{TV06} (chapter 4) for further details. 
 
 \begin{fact}
 For $\F_q^n \ni \alpha \neq 0$, we have:
 \[
 \langle 1,\chi_{\alpha} \rangle = 0.
 \]
 \end{fact}
 
 \begin{fact}[Plancherel's Theorem]
 \label{fact:Plancherel}
 For functions $f,g:\F^n \rightarrow \C$, we have
 \[
 \langle f,g\rangle = \sum_{\alpha}\hat{f}(\alpha) \hat{g}(\alpha).
 \]
 \end{fact}
 
 \begin{fact}
 \label{fact:conv}
 Suppose $g: \F_q^n \rightarrow \C$ can be written as a product $g(x) = \prod_{i \in [t]}g_i(x)$, then we have the Fourier coefficients of $g$ given by: 

\begin{align*}
\widehat{g}(\alpha) & = \left(\widehat{g_1}\ast \cdots \ast \widehat{g_t}\right)(\alpha)\\
& = \sum_{\beta_1,\ldots ,\beta_{t-1}} \widehat{g_1}(\beta_1)\cdots \widehat{g_{t-1}}(\beta_{t-1}) \widehat{g_t}(\alpha - \sum_{i \in [q-1]}\beta_i).
\end{align*}
\end{fact}

\begin{fact}
\label{prop:linspace}
If $g : \F_q^n \rightarrow \C$ is the indicator of a linear space $\mathcal{C}$, we have:

\[
    \widehat{g}(\alpha)= 
\begin{cases}
    \frac{|\mathcal{C}|}{|\F|^n},& \text{if } \alpha \in \mathcal{C}^{\perp}\\
    0,              & \text{otherwise}.
\end{cases}
\]
\end{fact}

\subsection{Hypercontractivity and sharp thresholds}
\label{sec:FK}

Here we state some tools from the analysis of Boolean function that we will use:

\begin{definition}
\label{defn:transitive}
We say that a function $f :\{0,1\}^n\rightarrow \{0,1\}$ is \emph{transitive-symmetric} if for every $i,j \in [n]$, there is a permutation $\sigma \in \mathfrak{S}_n$ such that:
\begin{itemize}
\item[1.] $\sigma(i) = j$
\item[2.] $f(x_{\sigma(1)},\ldots,x_{\sigma_{n}}) = f(x)$ for all $x \in \{0,1\}^n$.
\end{itemize}
\end{definition}

Let $f : \{0,1\} \rightarrow \{0,1\}$ be a monotone function. We will state an important theorem by Friedgut and Kalai, as stated in the excellent reference~\cite{O14}, regarding sharp thresholds for balanced symmetric monotone Boolean functions. This will be another important tool that we will use.

\begin{theorem}[~\cite{FK96}]
\label{thm:FK}
Let $f: \{0,1\}^n\rightarrow \{0,1\}$ be a nonconstant, monotone, transitive-symmetric function and let $F:[0,1]\rightarrow[0,1]$ be the strictly increasing function defined by $F(p)=\Prob(f(p)=1)$. Let $p_{\text{crit}}$ be the critical probability such that $F(p_{\text{crit}})=1/2$ and assume without loss of generality that $p_{\text{crit}}\leq1/2$. Fix $0<\epsilon<1/4$ and let
\[
\eta=B\log(1/\epsilon) \cdot \frac{\log(1/p_{\text{crit}})}{\log n},
\]

where $B>0$ is a universal constant. Then assuming $\eta\leq1/2$,
\[
F(p_{\text{crit}} \cdot(1 - \eta))\leq\epsilon,~~~F(p_{\text{crit}}\cdot(1+\eta)) \geq1-\epsilon.
\]
\end{theorem}

We will use an immediate corollary of the above theorem.

\begin{corollary}
\label{corr:FK}
Let $f:\{0,1\}^n\rightarrow \{0,1\}$ be a nonconstant, monotone, transitive-symmetric function. Let $F:[0,1]\rightarrow[0,1]$ be the strictly increasing function defined by $F(p)=\Pr(f(p)=1)$. Let $p$ be such that $F(p) \geq \epsilon$, and let  $\eta = B\log(1/\epsilon) \cdot \frac{\log(1/p_c)}{\log n}$. Then $F(p(1 + 2\eta)) \geq 1 - \epsilon$.
\end{corollary}

In particular, in the above corollary, if for some $\epsilon \in (0,1)$ we have that $F^{-1}(\epsilon) \in (0,1)$, then the function $f$ has a sharp threshold. 

One easy observation that will allow us to use Theorem~\ref{thm:FK} is the following:

\begin{proposition}
\label{prop:symm}
The Boolean function $\text{AP}_{r}:\{0,1\}^{q \times n} \rightarrow \{0,1\}$ is transitive-symmetric.
\end{proposition}

\begin{proof}
For a pair of coordinates indexed by $(i_1,j_1)$ and $(i_2,j_2)$, it is easy to see that the map $(x,y) \mapsto (x + i_2 - i_1,y+j_2-j_1)$ gives us what we need since the set of polynomials is invariant under these operations.
\end{proof}

\section{Proof of Theorem~\ref{thm:sharpthreshold}}
\label{sec:sharpthreshold}

First, we restate the theorem that we will prove in order to make a few more remarks:

\begin{theorem*}[Theorem~\ref{thm:sharpthreshold} restated]
Let $q$ be a prime power, and $r = r(q)$ and $\epsilon = \epsilon(q)$ be such that $q^{-r} \geq \frac{\log q}{q}$, and let $\epsilon \geq \omega\left(\max\left\{q^{-r}, \sqrt{q^{r - 1}  \log \left(q^{1 - r}\right)}\right\}\right)$.

Let $A_1,\ldots,A_q$ be independently chosen random subsets of $\F$ with each point picked independently with probability least $q^{-r}(1 + \epsilon)$. Then 
\[
\Prob((A_1 \times \cdots \times A_q) \cap \RS[q,rq] = \emptyset) = o(1).
\]
\end{theorem*}

\paragraph{Remarks} Before we proceed to the proof, we first make some simple observations that hopefully make the technical conditions on $q,c,\epsilon$ above seem more natural.

\begin{enumerate}
\item We need $q$ to be a prime power for the existence of $\F_q$.
\item If $r$ is too large, i.e., if $q^{-r} \leq \frac{\log q}{q}(1 - \delta)$, for some $\delta>0$, then we will almost surely not contain \emph{any} codeword. Indeed, we will almost surely have some $i \in [n]$ such that $A_i = \emptyset$.
\item The reason for $\epsilon = \sqrt{q^{r - 1}  \log \left(q^{1 - r}\right)}$ is more or less the same reason as above in that this helps us prove that w.h.p., $|A|$ is not much smaller than expected, as in Claim~\ref{claim:mainterm}. This is probably not the best dependence possible, and we make no attempt to optimize. But as $q^{-r}$ gets closer to  $\frac{\log q}{q}$, then this condition gets closer to the truth.
\end{enumerate}

We now proceed to the proof.

\begin{proof}
Let us abbreviate $\F = \F_q$, denote the subspace $\mathcal{C} \leq \F^q$ to be the $\RS[q,rq]$ code. Let $f:\F^q \rightarrow \C$ be the indicator of $\mathcal{C}$, i.e., $f(x) = 1$ iff $x \in \mathcal{C}$ and $0$ otherwise. Let $A_i \subset \F$ for $i \in [n]$ and let $g_i : \F \rightarrow \{0,1\}$ be the indicator for $A_i$. Let us slightly abuse notation and also think of $g_i:\F^q \rightarrow \C$ which depends only on the $i$'th variable.

We will estimate the quantity $|\mathcal{C} \cap \left(A_1 \times \cdots \times A_q \right)| = q^q\langle f, g \rangle$. Setting $|\widehat{f}(0)| =: \rho$, standard steps yield:

\begin{align}
\rho^{-1}\langle f, g \rangle & = \rho^{-1}\sum_{\alpha}\widehat{f}(\alpha) \widehat{g}(\alpha)  \nonumber \\
& = \sum_{\alpha \in \mathcal{C}^{\perp}}\widehat{g}(\alpha) \nonumber \\
& = \sum_{\alpha \in \mathcal{C}^{\perp}}\sum_{\beta_1,\ldots ,\beta_{q-1}} \widehat{g_1}(\beta_1)\cdots \widehat{g_{n-1}}(\beta_{q-1}) \widehat{g_q}(\alpha - \sum_{i \in [q-1]}\beta_i) \nonumber \\
& = \sum_{\beta_1,\ldots ,\beta_{q-1}} \widehat{g_1}(\beta_1)\cdots \widehat{g_{q-1}}(\beta_{q-1})\sum_{\alpha \in \mathcal{C}^{\perp}} \widehat{g_q}(\alpha - \sum_{i \in [q-1]}\beta_i) \nonumber \\
& = \sum_{(\alpha_1,\ldots, \alpha_q) \in \mathcal{C}^{\perp}}\prod_{i \in [q]}\widehat{g_i}(\alpha_i) \nonumber \\
&  \geq \prod_{i \in [q]}\widehat{g_i}(0) - \left|\sum_{(\alpha_1,\ldots, \alpha_q) \in \mathcal{C}^{\perp} \setminus \{\overline{0}\}} \left(\prod_{i \in [q]}\widehat{g_i}(\alpha_i) \right) \right|. \label{eqn:main}
\end{align}

where the first equality is due to Plancherel's identity, the third inequality is using Fact~\ref{fact:conv}, the and last equality is because of the fact that $\widehat{g_i}(\beta_i)$ is nonzero only if $\operatorname{supp}(\beta_i) \subseteq \{i\}$.  Let us denote 

\[
R := \sum_{(\alpha_1,\ldots, \alpha_q) \in \mathcal{C}^{\perp} \setminus \{\overline{0}\}} \left(\prod_{i \in [q]}\widehat{g_i}(\alpha_i) \right).
\]

For $\alpha \in \F^q$, let us define $M(\alpha) = \prod_{i \in [n]}\widehat{g_i}(\alpha_i)$, to be the `monomial' corresponding to $\alpha$. So, we have $R = \sum_{\mathcal{C}^{\perp} \setminus \{\overline{0}\}}M(\alpha)$, and $|R|^2 = \sum_{\alpha,\beta \in \mathcal{\C}^{\perp} \setminus \{\overline{0}\}}M(\alpha) \overline{M(\beta)}$. By linearity of expectation, we have

\begin{align*}
\E[|R|^2] & = \sum_{\alpha, \beta \in \mathcal{C}^{\perp} \setminus \{\overline{0}\}} \E[M(\alpha)\overline{M(\beta)}] \\
&= \sum_{\alpha \in \mathcal{C}^{\perp} \setminus \{\overline{0}\}}\E[M(\alpha)\overline{M(\alpha)}] + \sum_{\alpha \neq \beta}\E[M(\alpha)\overline{M(\beta)}].
\end{align*}

For $\alpha \neq \beta$, let $t$ be a coordinate such that $\alpha_t \neq \beta_t$. We have:

\begin{align*}
q^2 \cdot \E[\widehat{g_i}(\alpha_i) \overline{\widehat{g_i}(\beta_i)}] & = \E\left[\left(\sum_{x \in \F} g_i(x)\chi_{\alpha_i}(x)\right)\left(\overline{\sum_{y \in \F} g_i(x)\chi_{\beta_i}(y)}\right) \right] \\
& = \E\left[\sum_{x,y}g_i(x)g_i(y) \chi_{\alpha_i}(x) \overline{\chi_{\beta_i}(y)} \right] \\
& = \sum_{x,y}\left(\E\left[g_i(x)g_i(y) \chi_{\alpha_i}(x) \overline{\chi_{\beta_i}(y)}  \right]\right) \\
& = \sum_x \left(\E\left[g_i(x) \ \chi_{\alpha_i}(x) \overline{\chi_{\beta_i}(x)}  \right]\right) + \sum_{x \neq y} \left(\E\left[g_i(x)g_i(y)  \chi_{\alpha_i}(x) \overline{\chi_{\beta_i}(y)} \right]\right) \\
& = p \cdot \sum_{x}\left( \chi_{\alpha_i - \beta_i}(x) \right) + p^2 \cdot \sum_{x \neq y} \left( \chi_{\alpha_i}(x) \overline{\chi_{\beta_i}(y)}  \right) \\
& = (p - p^2)\cdot \sum_{x}\left( \chi_{\alpha_i - \beta_i}(x) \right) + p^2 \cdot \sum_{x , y} \left( \chi_{\alpha_i}(x) \overline{\chi_{\beta_i}(y)} \right) \\
& = (p - p^2)\cdot \sum_{x}\left( \chi_{\alpha_i - \beta_i}(x) \right)  + p^2 \left(\sum_x \chi_{\alpha_i}(x) \right)\left(\overline{\sum_y\chi_{\beta_i}(y)}\right)\\
& = 0.
\end{align*}

The last equality is because at least one of $\alpha_i$ or $\beta_i$ is nonzero, and $\alpha_i - \beta_i$ is nonzero. Therefore, for $\alpha \neq \beta$, we have:

\begin{align*}
\E[M(\alpha)\overline{M(\beta)}] & = \E\left[\left(\prod_{i}\widehat{g_i}(\alpha_i)\right)\left(\overline{\prod_i \widehat{g_i}(\beta_i)} \right) \right] \\
& = \prod_i \left( \E[\widehat{g_i}(\alpha_i) \overline{\widehat{g_i}(\beta_i)}] \right) \\
& = 0.
\end{align*}

Where the second equality is because $A_1,\ldots A_q$ are chosen independently. For $\alpha = \beta$, it is easy to see that $\E[M(\alpha)\overline{M(\alpha)}] = \prod_{i \in [q]}|\widehat{g_i}(\alpha_i)|^2$. So, we have the identity:

\begin{align}
\E[|R|^2] & = \sum_{\alpha \in \mathcal{C}^{\perp} \setminus \{0\}}\E\left[\prod_{i \in [q]}|\widehat{g_i}(\alpha_i)|^2\right] \nonumber \\
& = \sum_{\alpha \in \mathcal{C}^{\perp} \setminus \{0\}}\prod_{i \in [q]}\E\left[|\widehat{g_i}(\alpha_i)|^2\right]. \label{eqn:ER^2}
\end{align}

The following two identities are easy to check:

\begin{align}
\E\left[|\widehat{g_i}(0)|^2\right] & = p^2 + \frac{p(1-p)}{q} \label{eqn:Eg0}\\
\E\left[|\widehat{g_i}(\alpha_i)|^2\right] & = \frac{p(1-p)}{q} \text{ for } \alpha_i \neq 0.  \label{eqn:Egalpha}
\end{align}

Equipped with \ref{eqn:ER^2}, \ref{eqn:Eg0}, \ref{eqn:Egalpha}, and Proposition~\ref{prop:RSwtdist}, we have:

\begin{align*}
\E\left[|R|^2\right] & = \sum_{i = 0}^{q-1}W_{q-i}p^{2i}\left(1 + \frac{1-p}{pq}\right)^{2i}p^{q - i}\left(\frac{1-p}{q}\right)^{q-i} \\
& \leq \cdot \left(\frac{1 - p}{q} \right)^{q} \cdot \sum_{i = 0}^{q-1}\frac{q^{q - rq}}{i!}{p^{q+i}} \left( \frac{q}{1 - p}\right)^i\left( 1 + \frac{1-p}{pq}\right)^{2i}\\
& = \cdot \left(\frac{1 - p}{q} \right)^{q} \cdot q^{q-rq }p^{q} \sum_{i = 0}^{q-1}\frac{1}{i!}\left(\frac{pq}{1 - p} \right)^i \left( 1 + \frac{1-p}{pq}\right)^{2i}\\
& \leq e \cdot(1 - p)^q \cdot q^{q- rq }\left(\frac{p}{q}\right)^{q}e^{\left( \frac{2pq}{1 - p}\right)}.
\end{align*}

Markov's inequality gives us:

\begin{equation}
\label{eqn:Rterm}
\Prob\left(|R| \geq \left(qe \cdot q^{q-rq }\left(\frac{p}{q}\right)^{q}e^{\left( \frac{2pq}{1 - p}\right)}(1 - p)^q\right)^{\frac{1}{2}}\right) \leq \frac{1}{q}.
\end{equation}

On the other hand, we have:

\begin{claim}
\label{claim:mainterm}
For $q,c$ as given above, let $\epsilon = \omega\left(\sqrt{q^{r -1} \log \left(q^{1- r}\right)} \right)$, and $p = q^{-r}(1 + \epsilon)$. Then we have:

\[
\Prob\left(\prod_{i \in [q]}\widehat{g_i}(0) \leq q^{-rq}(1 + 0.9\epsilon)^{0.9q}\right) = o(1).
\]
\end{claim}

The proof is postponed to the Appendix.

So, using~\ref{eqn:main}, Claim~\ref{claim:mainterm}, and~\ref{eqn:Rterm}, and setting $p = q^{-r}(1 + \epsilon)$ we have that with probability at least $1 - o(1)$,

\begin{align*}
\rho^{-1}\langle f,g \rangle & \geq q^{-rq}(1 + 0.9\epsilon)^{0.9q} -  q\cdot eq^{-rq}(1 + \epsilon)^{\frac{q}{2}}e^{\frac{2pq}{2(1-  p)}}(1 - p)^{\frac{q}{2}} \\
& \geq q^{-rq}(1 + \epsilon)^{\frac{q}{2}}\left((1 + (\epsilon/2))^{0.4q} - eq \cdot e^{\frac{2pq}{2(1 - p)}}(1 - p)^{\frac{q}{2}}\right).
\end{align*}

Where the inequality follows from the fact that for $x \in [0,1]$,
\[
\frac{(1 + 0.9x)^{0.9}}{(1+x)^{0.5}} \geq (1 + 0.5x)^{0.4} .
\]

It remains to check that if $\epsilon  = \omega\left(q^{-r} \right)$, then $\rho^{-1}\langle f,g \rangle > 0$, which completes the proof.
\end{proof}

\section{Proof of Theorem~\ref{thm:fullrange}}
\label{sec:fullrange}

Here, we address the case when $r  = O\left(\frac{1}{\log q}\right)$. In this case, we observe that Theorem~\ref{thm:sharpthreshold} does not give us a sharp threshold for the random list recovery since in this case, $p \in (0,1)$. However, this case can be handled by the second moment method and Theorem~\ref{thm:FK}.

\begin{proof}[Proof of Theorem~\ref{thm:fullrange}]
Let us use $X$ to denote the number of codewords of $\text{RS}[q,rq]$ contained in a randomly chosen set $S$. Linearity of expectation gives us

\begin{equation}
\label{eqn:EX}
\E[X] = q^{rq+1}p^q
\end{equation}

Again, we have that if $p = q^{-r}(1 - \epsilon)$, we have that $\E[X] = (1 - \epsilon)^q$. And so, by Union Bound, we have that with high probability $X = 0$. On the other hand, if $p = q^{-r}(1 + \epsilon)$, we will show that with high probability, $X > 0$. We start off by computing the second moment.

\begin{lemma}
\label{lem:easysecondmoment}
We have 
\[ 
\E[X^2] \leq \E[X] + e^{\frac{1}{p}}\E^2[X]
\]
\end{lemma}

\begin{proof}
Let us denote $\mathcal{C} = \RS[q,rq]$. For every $w \in \mathcal{C}$, let us denote $X_w$ for the indicator random variable for the event $\{w \subset S\}$. We have:
\begin{align*}
\E[X^2] & = \sum_{w \in \mathcal{C}}\E[X_w] + \sum_{\substack{ w_1,w_2 \in \mathcal{C} \\ w_1 \neq w_2}}\E[X_{w_1}X_{w_2}] \\
& = q^{rq+1}p^q + \sum_{i = 0}^{rq+1}\sum_{\substack{w_1,w_2 \in \mathcal{C}\\ \Delta(w_1,w_2) = q - i}}p^{2q - i} \\
& = q^{rq+1}p^q + q^{rq+1}\sum_{i = 0}^{n - rq}|W_{q-i}|p^{2q - i} \\
& \leq q^{rq+1}p^q + q^{rq+1}\sum_{i = 0}^{rq+1}\frac{q^{rq+1}}{i!}p^{2q - i} \\
& \leq q^{k}p^q + e^{\frac{1}{p}}\left( q^{rq+1} p^{q}\right)^2 \\
& = \E[X] + e^{\frac{1}{p}}\E^2[X].
\end{align*}
\end{proof}

When $\E[X] = \omega(1)$, one can bound second moment by $2e^{\frac{1}{p}}\E^2[X]$. The Paley-Zygmund inequality~\ref{eqn:PaleyZygmund} immediately gives:

\[
\Prob(X \geq 0) \geq \frac{1}{2e^{\frac{1}{p}}}.
\]

Now, Corollary~\ref{corr:FK} gives us that if $p \geq q^{-r}\left(1 + \omega_{p}\left(\frac{1}{\log q} \right)\right)$, then $\Prob(X > 0) = 1 - o(1)$.

This, combined with Lemma~\ref{lem:smallp} and Theorem~\ref{thm:sharpthreshold} finishes the proof.
\end{proof}

\section{Technical lemmas}
\label{sec:appendix}

\subsection{Proofs from Section~\ref{sec:sharpthreshold}}
\label{sec:asec3}
\begin{proof}[Proof of Claim~\ref{claim:mainterm}]
For $i \in [n]$,  let define the indicator random variables
\[
Y_i = \mathbbm{1}\left[\{\widehat{g_i}(0) < (1 + 0.9\epsilon)q^{-r}\}\right]
\] 
and 
\[
Z_i = \mathbbm{1}\left[\{\widehat{g_i}(0) = 0\}\right].
\]

Let $Y = \sum_iY_i$, and $Z = \sum_iZ_i$. First off, Chernoff bound gives us that there is a constant $C  < \frac{1}{200}$ such that $\Prob(Y_i = 1) \leq e^{- C\epsilon^2 pq}$, so we have $\E[Y] \leq ne^{- C\epsilon^2pq}$. Since all the $Y_i$'s are independent, Chernoff bound again gives us that
\begin{equation}
\label{eqn:A}
\Prob\left(Y \geq 2q e^{-C\epsilon^2 pq}\right) = o(1).
\end{equation}
Moreover, since we have $q^{-r} \geq \frac{\log q}{q}$, we have
\begin{equation}
\label{eqn:B}
\Prob(Z > 0) =  o(1)
\end{equation}

Let us abbreviate $t = 2q e^{-C\epsilon^2 pq}$. Taking~\ref{eqn:A} and~\ref{eqn:B} into consideration, we have that with probability at least $1 - o(1)$,

\begin{align*}
\prod_{i \in [q]}\widehat{g_i}(0) & \geq p^{q}(1 + 0.9\epsilon)^{q - t}e^{-t\log q} \\
& \geq (1 + 0.9\epsilon)^{q - t - \frac{2t \log q}{\epsilon}} 
\end{align*}

It remains to check that if $\epsilon \gg \sqrt{\frac{\log \left(pq\right)}{pq}}$, then $t + \frac{2t \log q}{\epsilon} \leq 0.1q$. With this in mind, we compute

\begin{align*}
t + \frac{2t \log q}{\epsilon} & \leq \frac{3t \log q}{\epsilon} \\
& \leq 6q \cdot \frac{1}{(pq)^{\omega(1)}}\cdot\sqrt{pq} \log q \\
& = o(q)
\end{align*}

which finishes the proof.

\end{proof}

\subsection{Proofs from Section~\ref{sec:AC0}}
\label{sec:asec4}

\begin{proof}[Proof of Lemma~\ref{lem:bias}]
Consider the sequence of integers $\{k_i\}_{i \in \N}$ such that for every $i$, $k_i$ is the largest such that 

\[
\prod_{j = 1}^i\left(1 - \frac{1}{2^j}\right)^{k_j} \geq p.
\]

We make a basic observation:

\begin{observation}
\label{obs:bias}
We have that $k_1  = \lfloor \log_2 (1/p) \rfloor \leq s$ and for all $j \geq 2$, we have that $k_j \leq 3$.
\end{observation}

Let $\ell$ be the largest such that  $k_{\ell} > 0$ and $\sum_{i \in [\ell]}i\cdot k_i < s^2$. Let $t = \sum_{i \in [\ell]}i\cdot k_i $. Consider the $CNF$ given by
\[
C_p = \bigwedge_{j \in [\ell]}\left( \bigwedge_{i \in k_j}C_i^j\right)
\]

where the clause $C_i^j$ is an $\lor$ of $j$ independent variables. We first estimate $p_{\frac{1}{2}} := \Prob(C_p = 1)$ when $X_1,\ldots, X_n \sim \Ber \left(\frac{1}{2}\right)$.

Using the fact that $k_1 \leq s$ and Observation~\ref{obs:bias}, we have $\ell = \Omega(\sqrt{t})$. Therefore
\[
p \leq \Prob(C_p = 1)  \leq p\left( 1 - \frac{1}{2^{{\ell}+1}}\right)^{-4} \leq p\left(1 + \frac{4}{2^{\Omega(\sqrt{t})}} \right).
\]

And so for $X_1,\ldots, X_n \sim \Ber \left(\frac{1}{2} + \epsilon \right)$ we bound

\begin{align*}
\Prob(C_p = 1) & = \prod_{j = 1}^\ell\left(1 - \left(\frac{1}{2} - \epsilon\right)^j\right)^{k_j}  ~~~~~~~~=  \prod_{j = 1}^\ell\left(1 - \frac{1}{2^j}\left(1 - 2\epsilon\right)^j\right)^{k_j} \\
& \geq   \prod_{j = 1}^\ell\left(1 - \frac{1}{2^j}\left(1 - \epsilon j\right)\right)^{k_j}  ~~~~~~~\geq \prod_{j = 1}^\ell\left(1 - \frac{1}{2^j}\right)^{k_j}\left(1 + (\epsilon j/2^{j-1})\right)^{k_j} \\
& \geq p_{\frac{1}{2}}(C_p = 1)\left(1 + \epsilon \sum_{j = 1}^{\ell}\frac{jk_j}{2^{j-1}} \right)  \geq p(1 + \epsilon k_1 ).
\end{align*}

Similarly,  for $X_1,\ldots, X_n \sim \Ber \left(\frac{1}{2} - \epsilon \right)$:

\begin{align*}
\Prob(C_p = 1) & = \prod_{j = 1}^\ell\left(1 - \left(\frac{1}{2} + \epsilon\right)^j\right)^{k_j} ~~=  \prod_{j = 1}^\ell\left(1 - \frac{1}{2^j}\left(1 + 2\epsilon\right)^j\right)^{k_j} \\
& \leq   \prod_{j = 1}^\ell\left(1 - \frac{1}{2^j}\left(1 + 4\epsilon j\right)\right)^{k_j}  \leq \prod_{j = 1}^\ell\left(1 - \frac{1}{2^j}\right)^{k_j}\left(1 - \frac{\epsilon j}{2^{j}}\right)^{k_j} \\
& \leq p_{\frac{1}{2}}(C_p = 1)(1 - \Omega(k_1 \epsilon ) )  ~~\leq p\left(1 + \frac{1}{2^{\Omega(\sqrt{t})}} - \Omega(\epsilon k_1)\right).
\end{align*}

\end{proof}

\end{document}